\documentclass[11pt,table]{article}
\usepackage[utf8]{inputenc}
\usepackage{inputenc}
\usepackage[T1]{fontenc}
\usepackage{amsmath}
\usepackage{amssymb}
\usepackage{tabularx}
\usepackage{array,diagbox,makecell}
\usepackage{upgreek}
\usepackage{graphicx}
\usepackage{authblk}
\usepackage{layout}
\usepackage{dsfont}
\usepackage{bbold}
\usepackage[english]{babel}
\usepackage[top=2.2cm, bottom=3.2cm, left=2.2cm, right=2.2cm]{geometry}
\usepackage{multicol}
\usepackage{amsthm}
\usepackage{mathrsfs}
\usepackage{tikz}
\usepackage{tikz-network}
\usepackage{hyperref}
\usepackage{subfig}
\usepackage{complexity}
\usepackage[margin=1cm]{caption}
\usepackage{enumitem}
\usepackage{fontawesome}
\usepackage{soul}
\usepackage{stmaryrd}
\usepackage{tikz-cd}

\usepackage[font=small,labelfont=bf,tableposition=top]{caption}
\DeclareCaptionLabelFormat{andtable}{#1~#2  \&  \tablename~\thetable}

\newtheorem{theorem}{Theorem}[section]
\newtheorem{corollary}[theorem]{Corollary}

\newtheorem{claim}{Claim}

\newtheorem{lemma}[theorem]{Lemma}

\def\cqedsymbol{\ifmmode$\lrcorner$\else{\unskip\nobreak\hfil
\penalty50\hskip1em\null\nobreak\hfil$\lrcorner$
\parfillskip=0pt\finalhyphendemerits=0\endgraf}\fi} 
\newcommand{\cqed}{}

\newcommand{\Ga}{\mathcal{G}}
\newcommand{\strat}{\mathcal{S}}
\newcommand{\QBF}{\textsc{3-QBF}}
\newcommand{\true}{{\sf T}}
\newcommand{\false}{{\sf F}}
\newcommand{\quickset}[1]{\left\lbrace #1 \right\rbrace}

\newcommand{\segment}[2]{\llbracket #1, #2 \rrbracket}
\newcommand{\ind}{\textup{ind}}

\title{\LARGE Maker-Maker games of rank~$4$ are {\sf PSPACE}-complete}

\author[1]{\Large Florian Galliot}
\author[2]{\Large Jonas S\'enizergues}

\affil[1]{\large Aix-Marseille Université, CNRS, Centrale Marseille, I2M, UMR 7373, 13453 Marseille, France}
\affil[2]{\large Univ. Bordeaux, CNRS, Bordeaux INP, LaBRI, UMR 5800, F-33400 Talence, France}

\date{}

\begin{document}

\maketitle

\begin{abstract}
    The Maker-Maker convention of positional games is played on a hypergraph whose edges are interpreted as winning sets. Two players take turns picking a previously unpicked vertex, aiming at being first to pick all the vertices of some edge. Optimal play can only lead to a first player win or a draw, and deciding between the two is known to be {\sf PSPACE}-complete even for $6$-uniform hypergraphs. We establish {\sf PSPACE}-completeness for hypergraphs of rank~$4$. As an intermediary, we use the recently introduced achievement positional games, a more general convention in which each player has their own winning sets (blue and red). We show that deciding whether the blue player has a winning strategy as the first player is {\sf PSPACE}-complete even with blue edges of size~$2$ or~$3$ and pairwise disjoint red edges of size~$2$. The result for hypergraphs of rank~$4$ in the Maker-Maker convention follows as a simple corollary.
\end{abstract}

\section{Introduction}\strut
\indent\textbf{Positional games.} \textit{Positional games} have been introduced by Hales and Jewett \cite{Hales1963} and later popularized by Erd\H{o}s and Selfridge~\cite{erdos}. The game board is a hypergraph $H=(V,E)$, where $V$ is the vertex set and $E \subseteq 2^V$ is the edge set. Two players take turns picking a previously unpicked vertex, and the result of the game is defined by one of several possible \textit{conventions}. The study of positional games mainly consists, for various conventions and classes of hypergraphs, in finding necessary and/or sufficient conditions for such player to have a winning strategy (on the number of edges, the size of the edges, the structure of the hypergraph...), and determining the complexity of the corresponding algorithmic problems. The two most popular conventions are called \textit{Maker-Maker} and \textit{Maker-Breaker}. These are games of an ``achievement'' nature, in the sense that one or both players are trying to \textit{fill} any edge (\textit{i.e.} pick all the vertices of some edge of the hypergraph).

The Maker-Maker convention was the first one to be introduced in all generality, by Hales and Jewett in 1963 \cite{Hales1963}. It is the most natural convention: the edges are the winning sets, meaning that whoever first fills an edge wins the game. If no one has filled an edge by the time all vertices are picked, we get a draw. A well-known \textit{strategy-stealing} argument \cite{Hales1963} shows that optimal play leads to one of two possible outcomes: a first player win or a draw. The most famous example of a Maker-Maker game is the \textit{tic-tac-toe} game on a $3 \times 3$ grid, which can be generalized into the \textit{$k$-in-a-row} game on grids of any size\cite{Bec08}. For instance, the $k$-in-a-row game on an infinite grid is known to be a first player win for $k \leq 4$ but a draw for $k \geq 8$ \cite{Zetters}, while the cases $k \in \{5,6,7\}$ are important open problems. In the Maker-Maker convention, both players must manage offense and defense at the same time, by trying to fill an edge while also preventing the opponent from doing so first. This creates counterintuitive phenomena such as the \textit{extra edge paradox} \cite{Bec08}, where adding an edge might turn a first player win into a draw. As such, the Maker-Maker convention is notoriously difficult to handle.

The Maker-Breaker convention was introduced for that reason, by Chv\'atal and Erd\H{o}s in 1978 \cite{CE78}. The two players are called Maker and Breaker: Maker's goal is to fill an edge, while Breaker's goal is to prevent Maker from filling an edge. Draws are impossible. The fact that the players have complementary goals brings convenient additional properties compared to Maker-Maker games, and as such, the Maker-Breaker convention is the most studied. A crucial property is subhypergraph monotonicity: if Maker has a winning strategy on a subhypergraph of $H$, then he may disregard the rest of the hypergraph and apply that strategy to win on $H$, as he has no defense to take care of. An example of a Maker-Breaker game is the board game \textit{Hex}, where two players take turns placing tokens of their color to try and connect opposite sides of the board. It is well-known \cite{gardner} that a draw is technically impossible, and that successfully connecting one's sides of the board is equivalent to blocking the opponent from connecting theirs. This makes Hex a Maker-Breaker game, but only by theorem. Purely by definition, it does not fall under either Maker-Maker or Maker-Breaker conventions, as the first player to fill an edge wins but the players have different edges to fill (corresponding here to ``horizontal paths'' and ``vertical paths'' respectively).

\bigskip

\indent\textbf{Unified achievement games.} The more general family of \textit{achievement positional games} has recently been introduced \cite{JonasFlorian}. Such a game is a triple $\Ga=(V,E_L,E_R)$, where $(V,E_L)$ and $(V,E_R)$ are hypergraphs which are seen as having \textit{blue edges} and \textit{red edges} respectively. There are two players, Left and Right, taking turns picking a previously unpicked vertex: Left aims at filling a blue edge, while Right aims at filling a red edge. Whoever reaches their goal first wins the game, or we get a draw if this never happens. Achievement positional games include all Maker-Maker and Maker-Breaker games. Indeed, Maker-Maker games correspond to the case $E_L=E_R$, while Maker-Breaker games correspond to the case $E_R=\varnothing$ when identifying Maker with Left and Breaker with Right and renaming Breaker wins as draws. Another way to embed Maker-Breaker games into achievement positional games is to define $(V,E_R)$ as the transversal hypergraph of $(V,E_L)$, meaning that $E_R$ is the set of minimal subsets of vertices that intersect every element of $E_L$.

Some principles which are common to the Maker-Maker and Maker-Breaker conventions, such as strategy stealing and pairing strategies, generalize to achievement positional games \cite{JonasFlorian}. One of the motivations behind the introduction of achievement positional games, which is particularly relevant to our paper, is that the family of Maker-Maker games is not stable under the players' moves. Indeed, even though the winning sets are the same for both players at the start, any edge which contains a vertex picked by some player during the game can thereafter only be filled by that same player. As such, mid-game positions of Maker-Maker games fall under the wider category of achievement positional games.

\bigskip

\indent\textbf{Algorithmic complexity.} The decision problem corresponding to the Maker-Maker convention asks whether the first player has a winning strategy, while the decision problem corresponding to the Maker-Breaker convention asks whether Maker has a winning strategy as the first player. The algorithmic complexity of positional games is usually studied depending on the size of the edges. A hypergraph is of \textit{rank}~$k$ (resp. is \textit{$k$-uniform}) if all its edges have size at most~$k$ (resp. exactly~$k$).

The founding result of Schaefer states that Maker-Breaker games of rank~$11$ are {\sf PSPACE}-complete \cite{schaefer}. This was improved much later to $6$-uniform Maker-Breaker games \cite{MBrank6}, and very recently to 5-uniform Maker-Breaker games \cite{MBrank5}. On the other hand, Maker-Breaker games of rank~$3$ are solved in polynomial time \cite{MBrank3}.

As for Maker-Maker games, apart from the obvious remark that games of rank~$2$ are solved in polynomial time, no direct complexity results have been obtained. However, a simple argument due to Byskov \cite{byskov} reduces Maker-Breaker games of rank~$k$ to Maker-Maker games of rank~$k+1$. It is not difficult to adapt this construction to reduce $k$-uniform Maker-Breaker games to $(k+1)$-uniform Maker-Maker games. Therefore, the known results for the Maker-Breaker convention imply that $6$-uniform Maker-Maker games are {\sf PSPACE}-complete. However, using Maker-Breaker games to get complexity results on Maker-Maker games has its limitations: for instance, since Maker-Breaker games of rank~$3$ are tractable, it would be impossible to show that Maker-Maker games of rank~$4$ are {\sf PSPACE}-complete with this method.

Achievement positional games offer an alternative approach, as any starting position of an achievement positional game can be interpreted as a mid-game position of some Maker-Maker game (even more specifically, as a position obtained after just one round of play). The decision problem introduced in \cite{JonasFlorian} asks whether Left has a winning strategy as the first player. Table \ref{tab:results} sums up all previously known complexity results for that problem, depending on the sizes~$p$ and~$q$ of the blue and red edges respectively. In particular, the case $(p,q)=(3,3)$ being {\sf PSPACE}-complete implies that deciding the outcome of a Maker-Maker game of rank~$4$ after one round of play is {\sf PSPACE}-complete \cite{JonasFlorian}. Unfortunately, this does not settle the complexity of Maker-Maker games of rank~$4$, as this first round of play may not be optimal.

\begin{table}[h]
\begin{center}
\begin{tabular}{|c|c|c|c|c|c|} \hline
\diagbox{$q$}{$p$} & 0 , 1 & 2 & 3 & 4 & 5+ \\
\hline
0 , 1 & \makecell{{\sf LSPACE} \\{} \small [trivial]} & \makecell{{\sf LSPACE} \\{} \small \cite{RW20}} & \makecell{{\sf P} \\{} \small \cite{MBrank3}} & open & \makecell{{\sf PSPACE}-c \\{} \small \cite{MBrank5}} \\
\hline
2 & \makecell{{\sf LSPACE} \\{} \small [trivial]} & \makecell{{\sf P} \\{} \small \cite{JonasFlorian}} & \makecell{{\sf NP}-hard \\{} \small \cite{JonasFlorian}} & \makecell{{\sf NP}-hard \\{} \small \cite{JonasFlorian}} & \makecell{{\sf PSPACE}-c \\{} \small \cite{MBrank5}} \\
\hline
3+ & \makecell{{\sf LSPACE} \\{} \small [trivial]} & \makecell{{\sf coNP}-c \\{} \small \cite{JonasFlorian}} & \makecell{{\sf PSPACE}-c \\{} \small \cite{JonasFlorian}} & \makecell{{\sf PSPACE}-c \\{} \small \cite{JonasFlorian}} & \makecell{{\sf PSPACE}-c \\{} \small \cite{MBrank5}} \\
\hline
\end{tabular}
\end{center}
\caption{Previously known results for the algorithmic complexity of deciding whether Left has a winning strategy as the first player, with blue edges of size at most~$p$ and red edges of size at most~$q$.}\label{tab:results}
\end{table}


\indent\textbf{Our contribution.} We show that the case $(p,q)=(3,2)$ of achievement positional games is {\sf PSPACE}-complete, even when the red edges are restricted to be pairwise disjoint. Such a game can also be interpreted as the position obtained after one round of play in a Maker-Maker game of rank~$4$, but the major difference is that this round of play is actually optimal, so that our result implies {\sf PSPACE}-completeness for (starting positions of) Maker-Maker games of rank~$4$.

In Section \ref{section2}, after recalling some definitions and basic results from \cite{JonasFlorian}, we state the {\sf PSPACE}-completeness result for $(p,q)=(3,2)$ and we show that it implies {\sf PSPACE}-completeness for Maker-Maker games of rank~$4$. Section \ref{section3} is then dedicated to the proof of the result for $(p,q)=(3,2)$. Finally, Section \ref{section4} concludes the paper and lists some perspectives.

\section{Preliminaries and statements of the main results}\label{section2}

\subsection{Definitions}\strut
\indent We start by recalling some definitions from \cite{JonasFlorian}. In this paper, a \textit{hypergraph} is a pair $(V,E)$ where $V$ is a finite \textit{vertex set} and $E \subseteq 2^V \setminus \{\varnothing\}$ is the \textit{edge set}. An \textit{achievement positional game} is a triple $\Ga=(V,E_L,E_R)$ where $(V,E_L)$ and $(V,E_R)$ are hypergraphs. The elements of $E_L$ are called \textit{blue edges}, whereas the elements of $E_R$ are called \textit{red edges}. Two players, Left and Right, take turns picking a vertex in $V$ that has not been picked before. We say a player \textit{fills} an edge if that player has picked all the vertices of that edge. The blue and red edges can be seen as the winning sets of Left and Right respectively, so that the result of the game is determined as follows:
\begin{itemize}[noitemsep,nolistsep]
    \item If Left fills a blue edge before Right fills a red edge, then Left wins.
    \item If Right fills a red edge before Left fills a blue edge, then Right wins.
    \item If none of the above happens before all vertices are picked, then the game is a draw.
\end{itemize}

The player who starts the game can be either Left or Right. Therefore, when talking about winning strategies (\textit{i.e.} strategies that guarantee a win) or non-losing strategies (\textit{i.e.} strategies that guarantee a draw or a win), we will always specify which player is assumed to start the game. For instance, we may say that ``Left has a winning strategy on $\Ga$ as the first player''.

After Left picks a vertex $u$, any blue edge $e$ that contains $u$ behaves like $e \setminus \{u\}$, in the sense that Left only has to pick the vertices in $e \setminus \{u\}$ to fill $e$. Moreover, after Left picks a vertex $u$, any red edge $e'$ that contains $u$ is ``dead'' and can be ignored for the rest of the game, as Right will never be able to fill $e'$. Of course, analogous observations can be made for Right. Therefore, from a starting achievement positional game $\Ga=(V,E_L,E_R)$ on which Left has picked a set of vertices $V_L$ and Right has picked a set of vertices $V_R$, we get a fresh new achievement positional game $\Ga'=(V',E'_L,E'_R)$ where:
\begin{align*}
    V' & = V \setminus (V_L \cup V_R); \\
    E'_L & = \{ e \setminus V_L \mid e \in E_L, e \cap V_R = \varnothing\}; \\
    E'_R & = \{ e \setminus V_R \mid e \in E_R, e \cap V_L = \varnothing\}.
\end{align*}
We may refer to $\Ga'$ as the \textit{updated game}, to the elements of $E'_L$ as the \textit{updated blue edges} and to the elements of $E'_R$ as the \textit{updated red edges}.

For example, let $\Ga=(\{b,b^1,b^2,b^3,b^4,b^5,b^6,d\},\{\{b,b^1,b^2\},\{b,b^2,b^3\},\{b,b^4,b^5\},\{b,b^5,b^6\}\},\{\{b,d\}\})$, as pictured in Figure \ref{fig:butterflybomb}. First, suppose that Left starts. Left can pick $b$ as his first move, which yields the updated game $(\{b^1,b^2,b^3,b^4,b^5,b^6,d\},\{\{b^1,b^2\},\{b^2,b^3\},\{b^4,b^5\},\{b^5,b^6\}\},\varnothing)$. Note that Left cannot lose anymore. By symmetry, assume that Right picks one of $b^1$, $b^2$ or $b^3$. Left picks $b^5$, thus ensuring to win the game with his next move by picking either $b^4$ or $b^6$. Now, suppose that Right starts. Right can pick $b$ as her first move, which yields the updated game $(\{b^1,b^2,b^3,b^4,b^5,b^6,d\},\varnothing,\{\{d\}\})$. Left is forced to pick $d$, otherwise Right would pick $d$ herself next and win. In conclusion, we see that Left has a winning strategy as the first player, but we get a draw if Right plays first.

\begin{figure}[h]
    \centering
    \includegraphics[scale=.45]{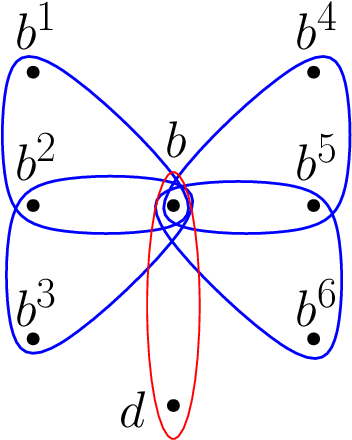}
    \caption{A blue \textit{butterfly} plus a red edge.}\label{fig:butterflybomb}
\end{figure}

\subsection{Elementary strategic principles}\strut
\indent We recall two basic results from \cite{JonasFlorian} which will be useful. A \textit{complete pairing} of a hypergraph $H=(V,E)$ is a set $\Pi$ of pairwise disjoint pairs of vertices such that, for all $e \in E$, there exists $\pi \in \Pi$ satisfying $\pi \subseteq e$.

\begin{lemma}[Pairing Strategy]\label{lem:pairing}\textup{\cite[Lemma 3.5]{JonasFlorian}}
    Let $\Ga=(V,E_L,E_R)$ be an achievement positional game. If the hypergraph $(V,E_R)$ (resp. $(V,E_L)$) admits a complete pairing, then Left (resp. Right) has a non-losing strategy on $\Ga$ both as the first player and as the second player.
\end{lemma}

\begin{lemma}[Greedy Move]\label{lem:greedy}\textup{\cite[Lemma 3.8]{JonasFlorian}}
    Let $\Ga=(V,E_L,E_R)$ be an achievement positional game. Suppose that there is no edge of size~$1$, but there is a blue (resp. red) edge $\{u,v\}$ such that, for all $e \in E_L \cup E_R: u \in e \implies v \in e$. Then it is optimal for Left (resp. Right) as the first player to start by picking $v$, which forces the opponent to answer by picking $u$.
\end{lemma}

\subsection{Statement of the main results}\strut
\indent In Section \ref{section3}, we will show the following result.

\begin{theorem}\label{theo:32-new}
    Deciding whether Left has a winning strategy as the first player on an achievement positional game with blue edges of size~$2$ or~$3$ and pairwise disjoint red edges of size~$2$ is {\sf PSPACE}-complete.
\end{theorem}

As an immediate consequence, we get {\sf PSPACE}-completeness for Maker-Maker games of rank~$4$.

\begin{corollary}\label{coro:makermaker4}
    Deciding whether the first player has a winning strategy for the Maker-Maker game on a hypergraph of rank~$4$ is {\sf PSPACE}-complete.
\end{corollary}

\begin{proof}[Proof of Corollary \ref{coro:makermaker4} assuming Theorem \ref{theo:32-new}]
    Membership in {\sf PSPACE} is well known for Maker-Maker games. Let $\Ga=(V,E_L,E_R)$ be an achievement positional game with blue edges of size~$2$ or~$3$ and pairwise disjoint red edges of size~$2$. We define the hypergraph $H=(V \cup \{u,v\}, E)$, where $u$ and $v$ are new vertices and:
    $$ E = \{e \cup \{u\} \mid e \in E_L \} \cup \{e \cup \{v\} \mid e \in E_R \} \cup \{\{u,v\}\}. $$
    Note that $H$ has rank~$4$, is efficiently computed from $\Ga$, and behaves the exact same as the achievement positional game $\Ga'=(V \cup \{u,v\},E,E)$: the first player has a winning strategy for the Maker-Maker game on $H$ if and only if Left has a winning strategy on $\Ga'$ as the first player. Therefore, to deduce the {\sf PSPACE}-hardness of Maker-Maker games of rank~$4$ from Theorem \ref{theo:32-new}, it suffices to show that Left has a winning strategy on $\Ga'$ as the first player if and only if Left has a winning strategy on $\Ga$ as the first player.
    
    \begin{itemize}
        \item First suppose that Left has a winning strategy on $\Ga$ as the first player. Then, playing first on $\Ga'$, Left can simply start by picking $u$, which forces Right to pick $v$ because $\{u,v\} \in E$. The updated game is precisely $\Ga$, so Left has a winning strategy as the first player.
        \item Conversely, suppose that Left has a winning strategy $\strat$ on $\Ga'$ as the first player. We claim that $\strat$ necessarily has Left picking $u$ as his first move. Indeed:
        \begin{itemize}[nolistsep,noitemsep]
            \item[--] Suppose that $\strat$ instructs Left to start by picking $v$. Then Right is forced to pick $u$. After that move, $E_R$ forms a complete pairing of the updated blue edges.
            \item[--] Now, suppose that $\strat$ instructs Left to start by picking some $w \not\in \{u,v\}$. Then Right has no forced move since $\{u,v\}$ is the only element of $E$ of size less than~$3$. Therefore, Right can pick $u$, which forces Left to pick $v$. If Right has no forced move, then $E_R$ forms a complete pairing of the updated blue edges. If Right has a forced move \textit{i.e.} there exists some (necessarily unique) $w'$ such that $\{v,w,w'\} \in E$, then Right picks $w'$, and $E_R \setminus \{\{w,w'\}\}$ forms a complete pairing of the updated blue edges after that move.
        \end{itemize}
        We see that, if $\strat$ does not instruct Left to pick $u$ as his first move, then Right can always obtain a non-losing pairing strategy by Lemma \ref{lem:pairing}, contradicting the fact that $\strat$ is winning for Left. In conclusion, $\strat$ instructs Left to start by picking $u$, which forces Right to pick $v$. Since the updated game after these two moves is precisely $\Ga$, we conclude that Left has a winning strategy on $\Ga$ as the first player. \qedhere
    \end{itemize}
\end{proof}

\section{Proof of Theorem \ref{theo:32-new}}\label{section3}\strut
\indent Membership in {\sf PSPACE} is straightforward, regardless of assumptions on the size of the edges \cite{JonasFlorian}. As for {\sf PSPACE}-hardness, we perform a reduction from the classic \QBF~decision problem which has been shown {\sf PSPACE}-complete by Stockmeyer and Meyer \cite{QBF}, or rather from its complement which is also {\sf PSPACE}-complete since {\sf PSPACE}={\sf coPSPACE} \cite{copspace}.

The problem \QBF~can be formulated in terms of the following game. A logic formula $\phi$ in CNF form with clauses of size exactly~$3$ is given, along with a prescribed numbering of its variables $x_1,y_1,\ldots,x_n,y_n$. Two players, Falsifier and Satisfier, take turns assigning truth values to the variables. We may assume that Falsifier goes first, setting $x_1$ to \true~or \false. Satisfier then sets $y_1$ to \true~or \false. Falsifier then sets $x_2$ to \true~or \false, and so on until a full valuation $\mu$ is built.
Satisfier wins the game if $\mu$ satisfies $\phi$, otherwise Falsifier wins. \\

\begin{tabularx}{0.95\textwidth}{|l @{} l @{} X|}
	\hline
	\multicolumn{3}{|l|}{\,\,\QBF} \\ \hline
	Input $\,$ & : \,\, & A logic formula $\phi$ in CNF form, with clauses of size exactly~$3$, and an even number of variables with a prescribed ordering $x_1,y_1,\ldots,x_n,y_n$. \\
	Output $\,$ & : \,\, & {\sf T} if Satisfier has a winning strategy as the second player, {\sf F} otherwise. \\ \hline
\end{tabularx} \\

Let $\phi$ be a logic formula in CNF form, with $m$ clauses $c_1,\ldots,c_m$ of size exactly~$3$, and $2n$ variables with a prescribed ordering $x_1,y_1,\ldots,x_n,y_n$. For all $j\in\segment{1}{m}$, write $c_j=\ell_j^1 \vee \ell_j^2 \vee \ell_j^3$, where the literals are ordered the same as their corresponding variables.
We are now going to define an achievement positional game $\Ga$, with blue edges of size~$2$ or~$3$ and pairwise disjoint red edges of size~$2$, on which Left has a winning strategy as the first player if and only if Falsifier has a winning strategy for \QBF~on $\phi$ as the first player.

\subsection{Idea of the construction}\strut
\indent Before going into details, let us briefly explain the general idea. We see Left (who starts) as Falsifier, and Right as Satisfier. Note that the red edges being pairwise disjoint implies that optimal play can only lead to a win for Left (\textit{i.e.} Falsifier wins) or a draw (\textit{i.e.} Satisfier wins), since Left could always apply a pairing strategy as per Lemma \ref{lem:pairing}. Each variable $x_i$ (resp. $y_i$) will be represented by two \textit{key vertices} $x_i^{\sf T}$ and $x_i^{\sf F}$ (resp. $y_i^{\sf T}$ and $y_i^{\sf F}$), and each clause will be represented by a blue butterfly.
With optimal play, the game $\Ga$ will have two phases. During Phase 1, the players' moves will be largely forced, but they will each have $n$ decisions to make, corresponding to the truth values they alternately choose as they build a full valuation $\mu$ for the variables. For $i$ from $1$ to $n$, Left will choose between picking $x_i^{\sf T}$ or $x_i^{\sf F}$ (with Right picking the other), hereby defining $\mu(x_i)$, then Right will choose between picking $y_i^{\sf T}$ or $y_i^{\sf F}$ (with Left picking the other), hereby defining $\mu(y_i)$.
By the end of Phase 1, Right will have destroyed every butterfly corresponding to a clause that is satisfied by $\mu$. During Phase 2, Right will hold on to a draw if $\mu$ satisfies $\phi$, whereas Left will use an intact blue butterfly to win if $\mu$ does not satisfy $\phi$.
There will be five types of edges:
\begin{itemize}[noitemsep,nolistsep]
    \item[--] \textit{Guide-edges} (blue and red) forcing both players' moves between two decisions they have to make during Phase 1;
    \item[--] \textit{Butterfly-edges} (blue only) forming butterflies which correspond to the clauses of $\phi$ and allow Left to win During Phase 2 if $\mu$ does not satisfy $\phi$;
    \item[--] \textit{Destruction-edges} (red only) allowing Right to destroy all the blue butterflies if $\mu$ satisfies $\phi$;
    \item[--] \textit{Trap-edges} (blue only) preventing Right from ``cheating'' during Phase 1 when she is supposed to choose between two key vertices to pick;
    \item[--] \textit{Link-edges} (blue only) connecting the clause gadgets to their associated variable gadgets.
\end{itemize}
The main challenge in the proof will be to verify that \textit{regular play}, which refers to the very restricted set of strategies that we define for both players during Phase 1, is actually optimal.

\subsection{Definition of the game}\strut
\indent We associate each literal $\ell$ with an index $\ind(\ell)$ defined as the unique $i \in \segment{1}{n}$ such that $\ell \in \{x_i,\neg x_i, y_i, \neg y_i\}$, as well as a key vertex $v(\ell)$ defined by: $v(x_i)=x_i^{\sf F}$, $v(\neg x_i)=x_i^{\sf T}$, $v(y_i)=y_i^{\sf T}$, $v(\neg y_i)=y_i^{\sf F}$. We define $\Ga=(V,E_L,E_R)$, where $V = \bigcup_{1 \leq i \leq n} V_i \cup \bigcup_{1 \leq j \leq m} C_j$, $E_L = \bigcup_{1 \leq i \leq n}(G_i^L \cup T_i) \cup \bigcup_{1 \leq j \leq m} (B_j \cup L_j)$ and $E_R = \bigcup_{1 \leq i \leq n}G_i^R \cup  \bigcup_{1 \leq j \leq m} D_j$, as follows.

\begin{itemize}

    \item For each $i \in \segment{1}{n}$, we have a variable gadget associated with the variables $x_i$ and $y_i$, with vertex set $V_i=\{x_i^{\sf T},x_i^{\sf F},y_i^{\sf T},y_i^{\sf F},w_i,w'_i,s_i,t_i,\omega_i,\omega'_i,\zeta_i,\zeta'_i,\tau_i,\delta_i,\lambda_i\}$. The variable gadgets only contain the guide-edges. We introduce the following useful notation for the list of blue guide-edges. Given two vertices $u,u'$ and an integer $i \in \segment{1}{n}$, the notation $\{u,u'\}^{(i)}$ means that we have the two blue edges $\{u,u',x_i^{\sf T}\}$ and $\{u,u',x_i^{\sf F}\}$. The idea is that Left will pick either $x_i^{\sf T}$ or $x_i^{\sf F}$, which creates the blue edge $\{u,u'\}$ in the updated game in both cases. We extend that notation to $i=0$ by setting $\{u,u'\}^{(0)}=\{u,u'\}$. The guide-edges are:
    \begin{align*}
    G_i^L & = \,\, \left\{\right. \{x_i^{\sf T},x_i^{\sf F}\}^{(i-1)},\{w_i,s_i\}^{(i)},\,\{w'_i,s_i\}^{(i)}, \\
    & \quad\quad\,\,\, \{y_i^{\sf T},y_i^{\sf F}\}^{(i)},\{y_i^{\sf T},\zeta_i\}^{(i)},\{y_i^{\sf F},\zeta'_i\}^{(i)},\{\zeta_i,\tau_i\}^{(i)},\{\zeta'_i,\tau_i\}^{(i)},\{\omega_i,\omega'_i\}^{(i)},\{\delta_i,\lambda_i\}^{(i)}\left.\right\}; \\
    G_i^R & = \,\, \left\{\right. \{x_i^{\sf T},w_i\},\{x_i^{\sf F},w'_i\},\{s_i,t_i\}, \\
        & \quad\quad\,\,\, \{y_i^{\sf T},\omega_i\},\{y_i^{\sf F},\omega'_i\},\{\zeta_i,\zeta'_i\},\{\tau_i,\delta_i\} \left.\right\}.
    \end{align*}
    Note that, even though we did not include $x_{i-1}^{\sf T}$ and $x_{i-1}^{\sf F}$ in our definition of $V_i$ to avoid repeating vertices, the variable gadget for $i \geq 2$ is actually connected to the variable gadget for $i-1$ through the blue edges $\{x_i^{\sf T},x_i^{\sf F},x_{i-1}^{\sf T}\}$ and $\{x_i^{\sf T},x_i^{\sf F},x_{i-1}^{\sf F}\}$. Figure \ref{fig:23-variable1} with $i=1$ illustrates the variable gadget associated with $x_1$ and $y_1$.

    \item For each $j\in\segment{1}{m}$, we have a clause gadget associated with the clause $c_j$, with vertex set $C_j = \{v(\ell_j^1),v(\ell_j^2),v(\ell_j^3),a_j^1,a_j^2,a_j^3,d_j,d_j^1,d_j^2,d_j^3,b_j,b_j^1,b_j^2,b_j^3,b_j^4,b_j^5,b_j^6\}$. Note that the key vertices are the only ones to appear in both variable gadgets and clause gadgets. The clause gadgets contain the butterfly-edges, destruction-edges and link-edges. The butterfly-edges are:
    \begin{align*}
        B_j & = \,\, \left\{\right. \{b_j,b_j^1,b_j^2\},\{b_j,b_j^2,b_j^3\},\{b_j,b_j^4,b_j^5\},\{b_j,b_j^5,b_j^6\} \left.\right\}.
    \end{align*}
    The destruction-edges are:
    \begin{align*}
        D_j & = \,\, \left\{\right. \{b_j,d_j\},\{b_j^1,d_j^1\},\{b_j^2,d_j^2\},\{b_j^3,d_j^3\} \left.\right\}.
    \end{align*}
    The link-edges are:
    \begin{align*}
        L_j & = \,\, \left\{\right. \{v(\ell_j^1),a_j^1,d_j^1\},\{v(\ell_j^2),a_j^2,d_j^2\},\{v(\ell_j^3),a_j^3,d_j^3\} \left.\right\}.
    \end{align*}

    As an illustration, the clause gadget associated with the clause $x_2 \vee \neg y_4 \vee y_7$ is pictured in Figure \ref{fig:23-clause1}.

    \begin{figure}[h]
        \centering
        \includegraphics[scale=.43]{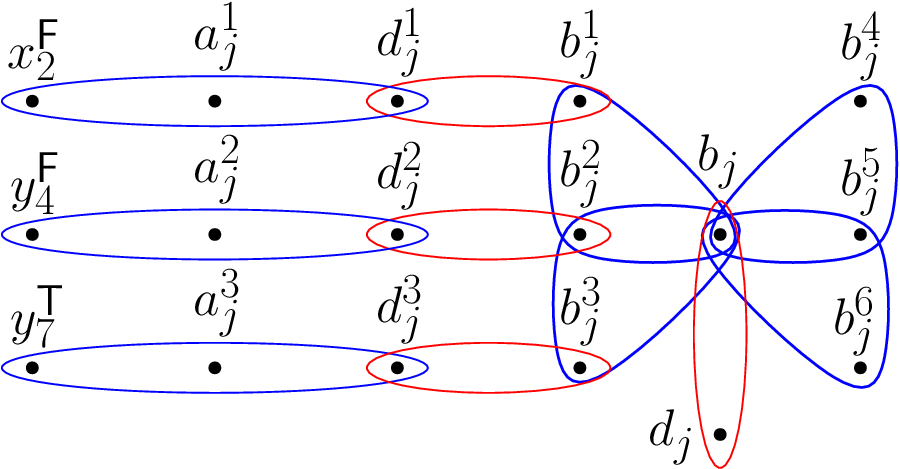}
        \caption{The clause gadget associated with the clause $c_j = x_2 \vee \neg y_4 \vee y_7$.}\label{fig:23-clause1}
    \end{figure}

    \item Finally, there are trap-edges that go across the different gadgets.
    For each $i \in \segment{1}{n}$, we define $T_i$ as the set of all $\{t_i,\tau_i,u\}$ such that $u$ belongs to some red edge $e$ which is either: a guide-edge in $G_k^R$ where $k\in\segment{i+1}{n}$, a destruction-edge of the form $\{b_j,d_j\}$ where $j\in\segment{1}{m}$, or a destruction-edge of the form $\{b_j^r,d_j^r\}$ where $j\in\segment{1}{m}$, $r\in\segment{1}{3}$ and $\ind(\ell_j^r) \geq i$. The idea is that Left will pick $t_i$ and Right will be forced to defuse the corresponding trap-edges by quickly picking $\tau_i$.
        
\end{itemize}

An important remark is that all blue edges have size~$3$ except for $\{x_1^{\sf T},x_1^{\sf F}\}$, the single blue edge of size~$2$.

\subsection{Regular play}\strut
\indent Assuming Left starts, we define \textit{regular play} on $\Ga$ during Phase 1, where players build a valuation $\mu$ for the variables $x_1,y_1\ldots,x_n,y_n$. An arrow labeled ``f'' signifies a forced move (when a player cannot win in one move but the opponent threatens to do so), while an arrow labeled ``g'' signifies a greedy move in the sense of Lemma \ref{lem:greedy}. For $i$ from $1$ to $n$, \textit{round~$i$} of regular play goes as follows, with the starting situation in the associated variable gadget shown in Figure \ref{fig:23-variable1}.

    \begin{figure}[h]
        \centering
        \includegraphics[scale=.43]{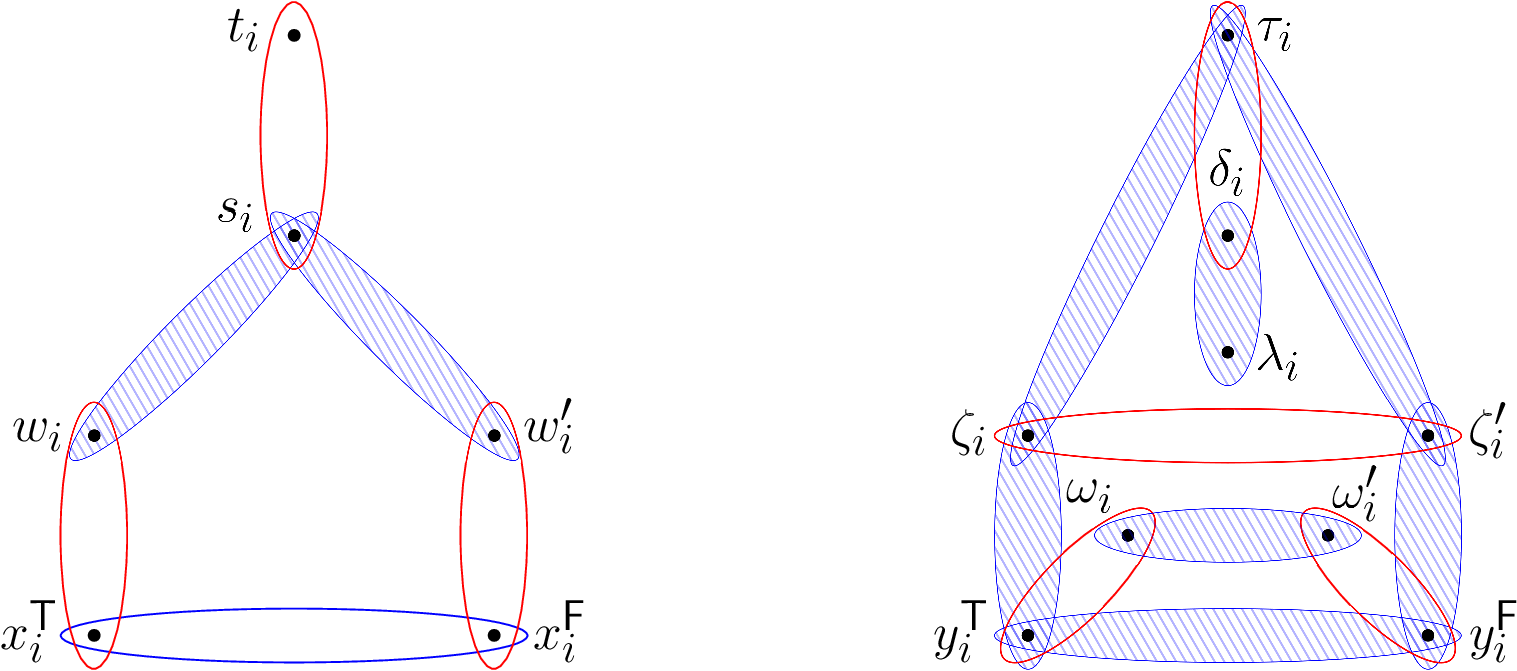}
        \caption{The updated variable gadget associated with $x_i$ and $y_i$ at the start of round~$i$ of regular play. For clarity, we use a shaded blue edge $\{u,u'\}$ to represent $\{u,u'\}^{(i)}$ (so that the left part and the right part are actually connected to each other).}\label{fig:23-variable1}
    \end{figure}

\begin{enumerate}[label={(\arabic*)}]

    \item Left chooses the value {\sf T} or {\sf F} for $x_i$, by picking $x_i^{\sf T}$ or $x_i^{\sf F}$ respectively. Both options trigger a sequence of forced moves inside the associated variable gadget:
        \begin{center}
        \begin{tikzcd}[row sep = tiny , sep = small]
            & \textcolor{blue}{x_i^{\sf T}} \arrow{r}{\text{f}} & \textcolor{red}{x_i^{\sf F}} \arrow{r}{\text{f}} & \textcolor{blue}{w'_i} \arrow{dr}{\text{f}} & & \\
            \arrow{ur}{\mu(x_i)={\sf T}} \arrow{dr}[swap]{\mu(x_i)={\sf F}} & & & & \textcolor{red}{s_i} \arrow{r}{\text{f}} & \textcolor{blue}{t_i} \\
            & \textcolor{blue}{x_i^{\sf F}} \arrow{r}{\text{f}} & \textcolor{red}{x_i^{\sf T}} \arrow{r}{\text{f}} & \textcolor{blue}{w_i} \arrow{ur}{\text{f}} & &
        \end{tikzcd}
        \end{center}

    In the updated game obtained after step~(1), we get the updated guide-edges pictured in Figure \ref{fig:23-variable2} (in the case where Left has put $x_i$ to {\sf T}, the other case being analogous).

    \begin{figure}[h]
        \centering
        \includegraphics[scale=.43]{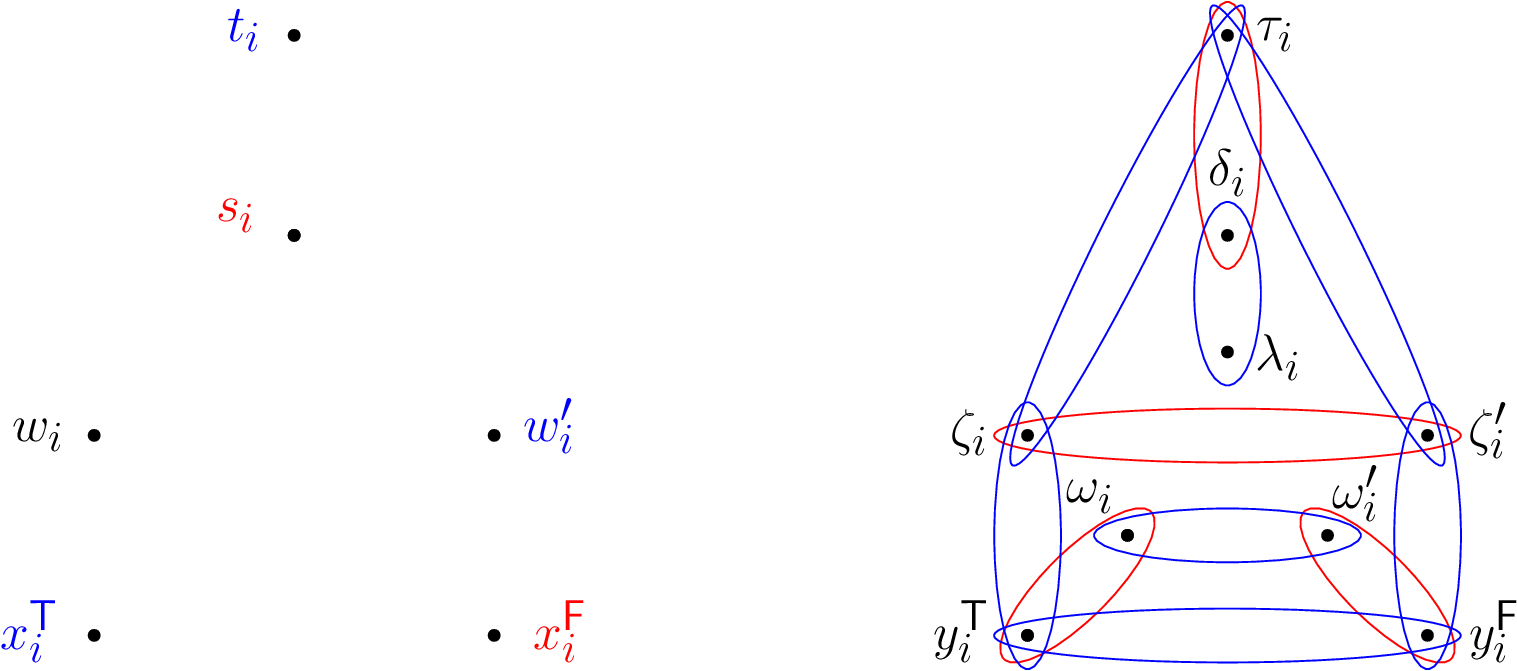}
        \caption{The updated variable gadget associated with $x_i$ and $y_i$ after step~(1) of round~$i$ of regular play (in the case where Left has put $x_i$ to {\sf T}).}\label{fig:23-variable2}
    \end{figure}

    \item Right now has no forced move, and uses this window to play greedy moves inside the clause gadgets, with the aim of destroying some blue butterflies (except if $i=1$, in which case she skips this step). For each $j\in\segment{1}{m}$ and each $r \in \segment{1}{3}$ such that Right has picked $v(\ell_j^r)$ during round~$i-1$ of regular play, Right plays the following greedy move:
        \begin{center}
        \begin{tikzcd}[row sep = tiny , sep = small]
            \arrow{r}{\text{g}} & \textcolor{red}{b_j^r} \arrow{r}{\text{f}} & \textcolor{blue}{d_j^r}
        \end{tikzcd}
        \end{center}
    For instance, consider $i=3$ and $c_j = x_2 \vee \neg y_4 \vee y_7$: if Right has picked $x_2^{\sf F}=v(\ell_j^1)$ during round~$2$ of regular play, then Right will play the greedy move $b_j^1$ as illustrated in Figure \ref{fig:23-clause2}. 

    \begin{figure}[h]
        \centering
        \includegraphics[scale=.43]{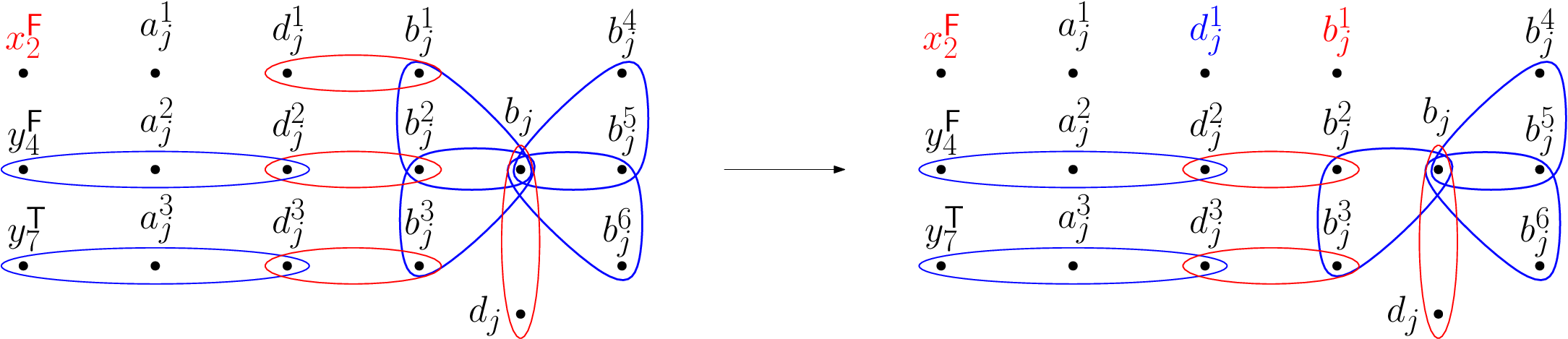}
        \caption{The clause gadget associated with the clause $c_j = x_2 \vee \neg y_4 \vee y_7$, before (left) and (after) Right's greedy move during round~$3$ of regular play if she had picked $x_2^{\sf F}$ in the previous round.}\label{fig:23-clause2}
    \end{figure}

    \item Right chooses the value {\sf T} or {\sf F} for $y_i$, by picking $y_i^{\sf T}$ or $y_i^{\sf F}$ respectively. Both options trigger a sequence of forced moves inside the associated variable gadget:
        \begin{center}
        \begin{tikzcd}[row sep = tiny , sep = small]
            & \textcolor{red}{y_i^{\sf T}} \arrow{r}{\text{f}} & \textcolor{blue}{\omega_i} \arrow{r}{\text{f}} & \textcolor{red}{\omega'_i} \arrow{r}{\text{f}} & \textcolor{blue}{y_i^{\sf F}} \arrow{r}{\text{f}} & \textcolor{red}{\zeta'_i} \arrow{r}{\text{f}} & \textcolor{blue}{\zeta_i} \arrow{dr}{\text{f}} & & & \\
            \arrow{ur}{\mu(y_i)={\sf T}} \arrow{dr}[swap]{\mu(y_i)={\sf F}} & & & & & & & \textcolor{red}{\tau_i} \arrow{r}{\text{f}} & \textcolor{blue}{\delta_i} \arrow{r}{\text{f}} & \textcolor{red}{\lambda_i} \\
            & \textcolor{red}{y_i^{\sf F}} \arrow{r}{\text{f}} & \textcolor{blue}{\omega'_i} \arrow{r}{\text{f}} & \textcolor{red}{\omega_i} \arrow{r}{\text{f}} & \textcolor{blue}{y_i^{\sf T}} \arrow{r}{\text{f}} & \textcolor{red}{\zeta_i} \arrow{r}{\text{f}} & \textcolor{blue}{\zeta'_i} \arrow{ur}{\text{f}} & & &
        \end{tikzcd}
        \end{center}
        
    \item Left now has no forced move, and uses this window to play greedy moves inside the clause gadgets, with the aim of removing some destruction-edges to protect her blue butterflies. For each $j \in \segment{1}{m}$ and each $r \in\segment{1}{3}$ such that Left has picked $v(\ell_j^r)$ during round~$i$ (the current round) of regular play, Left plays the following greedy move:
        \begin{center}
        \begin{tikzcd}[row sep = tiny , sep = small]
            \arrow{r}{\text{g}} & \textcolor{blue}{d_j^r} \arrow{r}{\text{f}} & \textcolor{red}{a_j^r}
        \end{tikzcd}
        \end{center}
    For instance, consider $i=2$ and $c_j = x_2 \vee \neg y_4 \vee y_7$: if Left has picked $x_2^{\sf F}=v(\ell_j^1)$ during round~$2$ of regular play, then Right will play the greedy move $d_j^1$ as pictured in Figure \ref{fig:23-clause3}. Round~$i$ of regular play is now over. Note that the updated variable gadget associated with $x_{i+1}$ and $y_{i+1}$, assuming $i<n$, is indeed as pictured in Figure \ref{fig:23-variable1} (replacing each $i$ by $i+1$) since Left picking either $x_i^{\sf T}$ or $x_i^{\sf F}$ during step~(1) has created the updated blue edge $\{x_{i+1}^{\sf T},x_{i+1}^{\sf F}\}$.

    \begin{figure}[h]
        \centering
        \includegraphics[scale=.43]{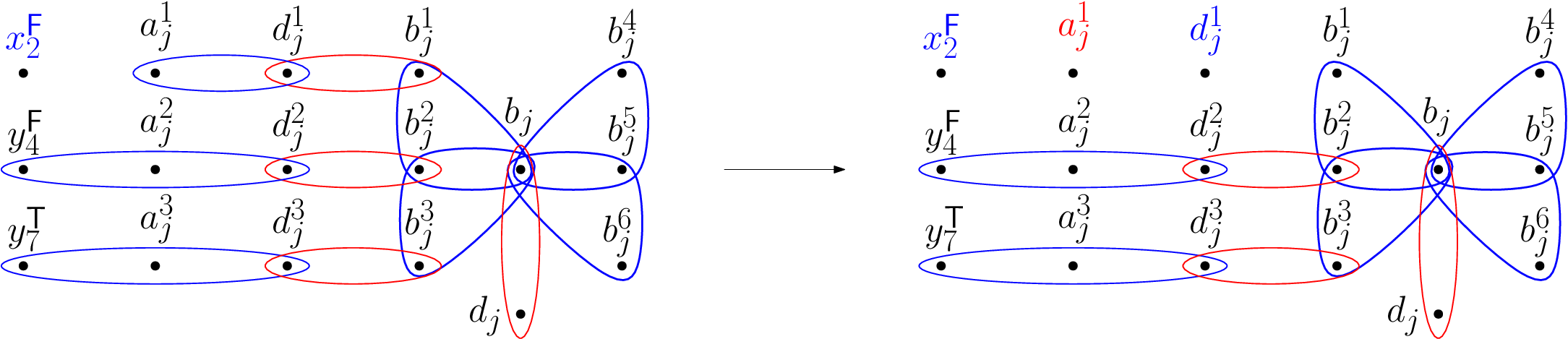}
        \caption{The clause gadget associated with the clause $c_j = x_2 \vee \neg y_4 \vee y_7$, before (left) and (after) Left's greedy move during round~$2$ of regular play if he had picked $x_2^{\sf F}$ at the start of the round.}\label{fig:23-clause3}
    \end{figure}
    
\end{enumerate}

Let us state some simple properties of regular play. We say a blue (resp. red) edge is \textit{dead} if Right (resp. Left) has \textit{killed} it by picking one of its vertices, or \textit{intact} if none of its vertices has been picked.

\begin{claim}\label{cla:regular0}
    Let $i\in\segment{1}{n}$, and consider the situation at the end of round~$i$ of regular play, assuming that both players have followed regular play from the beginning. The following properties hold:
    \begin{enumerate}[label={\textup{(\roman*)}}]
        \item All guide-edges in $\bigcup_{1 \leq k \leq i} (G_k^L \cup G_k^R)$ are dead. The other guide-edges are intact, apart (if $i<n$) from $\{x_{i+1}^{\sf T},x_{i+1}^{\sf F},x_i^{\sf T}\}$ and $\{x_{i+1}^{\sf T},x_{i+1}^{\sf F},x_i^{\sf F}\}$: one is dead, while the other has yielded $\{x_{i+1}^{\sf T},x_{i+1}^{\sf F}\}$ in the updated game.
        \item All trap-edges in $\bigcup_{1 \leq k \leq i} T_k$ are dead. The other trap-edges are intact.
        \item All link-edges $\{v(\ell_j^r),a_j^r,d_j^r\}$ such that $\ind(\ell_j^r) \leq i$ are dead. The other link-edges are intact.
        \item All destruction-edges of the form $\{b_j^r,d_j^r\}$ such that, either $\ind(\ell_j^r) \leq i$ and Left has picked $v(\ell_j^r)$, or $\ind(\ell_j^r) \leq i-1$ and Right has picked $v(\ell_j^r)$, are dead. The other destruction-edges are intact.
        \item Left has not picked any vertex in butterfly-edges. Moreover, for all $j\in\segment{1}{m}$, the butterfly-edges $\{b_j,b_j^4,b_j^5\}$ and $\{b_j,b_j^5,b_j^6\}$ are intact, and the two following assertions are equivalent:
            \begin{enumerate}[label={\textup{(\alph*)}},noitemsep,nolistsep]
                \item The butterfly-edges $\{b_j,b_j^1,b_j^2\}$ and $\{b_j,b_j^2,b_j^3\}$ are intact, and the three destruction-edges $\{b_j^1,d_j^1\}$, $\{b_j^2,d_j^2\}$ and $\{b_j^3,d_j^3\}$ are dead.
                \item Left has picked the three key vertices $v(\ell_j^1)$, $v(\ell_j^2)$ and $v(\ell_j^3)$.
            \end{enumerate}
        
    \end{enumerate}
\end{claim}

\begin{proof}[Proof of Claim~\ref{cla:regular0}]

\begin{enumerate}[label={\textup{(\roman*)}}]
    \item This is a direct consequence of the definition of regular play.
    \item Since Right has picked $\tau_k$ for all $k\in\segment{1}{i}$, all trap-edges in $\bigcup_{1 \leq k \leq i} T_k$ are dead. Now, consider some $\{t_{i'},\tau_{i'},u\} \in T_{i'}$ where $i'\in\segment{i+1}{n}$. Clearly, $t_{i'}$ and $\tau_{i'}$ are unpicked. Let $e$ be the red edge containing $u$. By definition of a trap-edge, there are three possibilities for $e$. If $e \in G_k^R$ for some $k\in\segment{i'+1}{n}$, then $u$ is unpicked since $k > i$. If $e$ is a destruction-edge of the form $\{b_j,d_j\}$, then $u$ is unpicked since those vertices are never picked during regular play. Finally, if $e$ is a destruction-edge of the form $\{b_j^r,d_j^r\}$ where $\ind(\ell_j^r) \geq i'$, then $u$ could only have been picked during step~(2) or (4) of some round of regular play, however this is impossible: since $\ind(\ell_j^r) > i$, the key vertex $v(\ell_j^r)$ has not yet been picked.
    \item Let $e=\{v(\ell_j^r),a_j^r,d_j^r\}$ be a link-edge. First, suppose that $\ind(\ell_j^r) \leq i$, \textit{i.e.} $v(\ell_j^r)$ has already been picked. If Right has picked $v(\ell_j^r)$, then $e$ is dead. If Left has picked $v(\ell_j^r)$, then Right has been forced to kill $e$ by picking $a_j^r$ during step~(4) of round~$\ind(\ell_j^r)$ of regular play. Finally, if $\ind(\ell_j^r) > i$, then $e$ is intact as $v(\ell_j^r)$ has not yet been picked.
    \item All destruction-edges of the form $\{b_j,d_j\}$ are intact since those vertices are never picked during regular play. Consider a destruction-edge of the form $e=\{b_j^r,d_j^r\}$. Assume $\ind(\ell_j^r) \leq i$, otherwise $e$ is clearly intact. If Left has picked $v(\ell_j^r)$, then Left has also killed $e$ by playing the greedy move $d_j^r$ during step~(4) of round~$\ind(\ell_j^r)$ of regular play. If Right has picked $v(\ell_j^r)$ and $\ind(\ell_j^r) \leq i-1$, then Left has been forced to kill $e$ by picking $d_j^r$ during step~(2) of round~$\ind(\ell_j^r)+1$ of regular play. If Right has picked $v(\ell_j^r)$ and $\ind(\ell_j^r) = i$, then $e$ is intact as the corresponding greedy move has not yet been played (it will during round~$i+1$ of regular play if $i<n$).
    \item The fact that Left has not picked any vertex in butterfly-edges and that the butterfly-edges $\{b_j,b_j^4,b_j^5\}$ and $\{b_j,b_j^5,b_j^6\}$ are intact for all $j\in\segment{1}{m}$ is a direct consequence of the definition of regular play. Now, let $j\in\segment{1}{m}$. If Left has picked $v(\ell_j^1)$, $v(\ell_j^2)$ and $v(\ell_j^3)$, then, by definition of regular play, Right has played no move inside the butterfly-edges in $B_j$ and Left has killed the destruction-edges $\{b_j^1,d_j^1\}$, $\{b_j^2,d_j^2\}$ and $\{b_j^3,d_j^3\}$ with greedy moves. If $v(\ell_j^r)$ is unpicked for some $r\in\segment{1}{3}$, then $\ind(\ell_j^r) > i$ so $\{b_j^r,d_j^r\}$ is intact. Finally, suppose that Left has picked $v(\ell_j^r)$ for some $r\in\segment{1}{3}$, which implies that $\ind(\ell_j^r) \leq i$. If $\ind(\ell_j^r) \leq i-1$, then Right has killed $\{b_j,b_j^1,b_j^2\}$ or $\{b_j,b_j^2,b_j^3\}$ by playing the greedy move $b_j^r$ during step~(2) of round~$\ind(\ell_j^r)+1$ of regular play. If $\ind(\ell_j^r)=i$, then $\{b_j^r,d_j^r\}$ is intact. \cqed \qedhere
\end{enumerate}
 
\end{proof}

We now show that regular play leads to the desired result for the game, \textit{i.e.} Left wins correspond to Falsifier wins.

\begin{claim}\label{cla:regular1}
    Assume both players are constrained to follow regular play during Phase 1. Then, Left has a winning strategy on $\Ga$ as the first player if and only if Falsifier has a winning strategy for \QBF~on $\phi$ as the first player.
\end{claim}

\begin{proof}[Proof of Claim~\ref{cla:regular1}]
    First, suppose that Falsifier has a winning strategy $\strat$ for \QBF~on $\phi$ as the first player. We define the following winning strategy for Left on $\Ga$ as the first player. During Phase 1, for all $i\in\segment{1}{n}$, Left chooses $\mu(x_i)$ depending on $\mu(x_1),\mu(y_1),\ldots,\mu(x_{i-1}),\mu(y_{i-1})$ and according to $\strat$, so that $\mu$ falsifies some clause $c_j$ of $\phi$. By definition of $\mu$, this means Left has picked $v(\ell_j^1)$, $v(\ell_j^2)$ and $v(\ell_j^3)$. Therefore, by item (v) of Claim~\ref{cla:regular0}, the edges in $B_j$ form an intact blue butterfly and all red edges intersecting that butterfly are dead apart from $\{b_j,d_j\}$. As Left is next to play, he picks $b_j$. After that, Right may stall by playing inside intact red edges (in which case Left picks the other vertex of the red edge each time), but at some point she will have to pick some $u$ that does not threaten a win in one move. By symmetry, assume $u \not\in \{b_j^4,b_j^5,b_j^6\}$. Left then picks $b_j^5$ and will win the game by picking $b_j^4$ or $b_j^6$ with his next move.

    Now, suppose that Satisfier has a winning strategy $\strat$ for \QBF~on $\phi$ as the second player.
    We define the following non-losing (actually, drawing) strategy for Right on $\Ga$ as the second player.
    During Phase 1, for all $i\in\segment{1}{n}$, Right chooses $\mu(y_i)$ depending on $\mu(x_1),\mu(y_1),\ldots,\mu(x_i)$ and according to $\strat$, so that $\mu$ ends up satisfying $\phi$. By Claim~\ref{cla:regular0} applied to $i=n$, all blue edges that are not butterfly-edges are dead. Moreover, for all $j\in\segment{1}{m}$, the clause $c_j$ being satisfied means that Right has picked at least one of $v(\ell_j^1)$, $v(\ell_j^2)$ or $v(\ell_j^3)$, so item (v) of Claim~\ref{cla:regular0} ensures that one of the butterfly-edges $\{b_j,b_j^1,b_j^2\}$ and $\{b_j,b_j^2,b_j^3\}$ is dead or some destruction-edge among $\{b_j^1,d_j^1\}$, $\{b_j^2,d_j^2\}$ and $\{b_j^3,d_j^3\}$ is intact. Also recall that the destruction-edge $\{b_j,d_j\}$ is intact for all $j\in\segment{1}{m}$ by Claim~\ref{cla:regular0}. As Left is next to play, he picks some $u$. Since Left had not picked any vertex in butterfly-edges before, Right's move is not forced. For all $j\in\segment{1}{m}$ (in any order) such that $u \not\in \{b_j,d_j\}$, Right picks $b_j$, forcing Left to pick $d_j$. If $u \not\in \bigcup_{1 \leq j \leq m} \{b_j,d_j\}$, then all blue edges are dead and the proof is over. Otherwise, let $j_0$ be the unique index such that $u \in \{b_{j_0},d_{j_0}\}$. Assume $u=b_{j_0}$, otherwise Right kills all remaining blue edges by picking $b_{j_0}$ herself. Recall that, either one of the butterfly-edges $\{b_{j_0},b_{j_0}^1,b_{j_0}^2\}$ and $\{b_{j_0},b_{j_0}^2,b_{j_0}^3\}$ was already dead by the end of Phase 1, or some destruction-edge $\{b_{j_0}^r,d_{j_0}^r\}$ was intact. In the latter case, Right picks $b_{j_0}^r$, forcing Left to pick $d_{j_0}^r$, so that one of the butterfly-edges $\{b_{j_0},b_{j_0}^1,b_{j_0}^2\}$ and $\{b_{j_0},b_{j_0}^2,b_{j_0}^3\}$ is now dead. Right picks $b_{j_0}^5$. At most one butterfly-edge is not dead at this point, and that edge has two unpicked vertices, so Right will kill it with her next move. \cqed
\end{proof}

\subsection{Optimality of regular play}\strut
\indent We now show that regular play is indeed optimal for both players. We start with Right, whose case is easier as her moves are hugely constrained by the trap-edges.

\begin{claim}\label{cla:regular2}
    Assume that Phase 1 is not over and that both players have followed regular play thus far. If Right is next to play, then it is optimal for Right to follow regular play with her next move.
\end{claim}

\begin{proof}[Proof of Claim~\ref{cla:regular2}]
    Let $\Ga'=(V',E'_L,E'_R)$ be the updated game, and let $i\in\segment{1}{n}$ be the current round of regular play. There are three types of moves for Right in regular play: forced move, greedy move, or ``decision move'' between $y_i^{\sf T}$ and $y_i^{\sf F}$.
    \begin{itemize}
        \item Let us first consider the case where regular play suggests a forced move. Clearly, Left is indeed threatening to win in one move (recall that, by picking $x_i^{\sf T}$ or $x_i^{\sf F}$ at the beginning of the round, Left has ``activated'' every edge of the form $\{u,u'\}^{(i)}$ into an actual blue edge of size~$2$: $\{u,u'\}$ in $\Ga'$). Therefore, the only way Right's move would not actually be forced is if Right could win in one move herself, but this is not the case since regular play has Left defending all of Right's threats.
        \item Now, we consider the case where regular play suggests a greedy move $b_j^r$ during step~(2), which implies that $\ind(\ell_j^r)=i-1$ and that Right has picked $v(\ell_j^r)$ during round~$i-1$ of regular play. Note that neither player is threatening a win in one move. Therefore, to show that this is a valid greedy move as per Lemma \ref{lem:greedy}, it suffices to show that $d_j^r \not\in e$ for all $e \in E'_L$. The only blue edges containing $d_j^r$ in the original game were the link-edge $\{v(\ell_j^r),a_j^r,d_j^r\}$ and some trap-edges. The link-edge $\{v(\ell_j^r),a_j^r,d_j^r\}$ is dead since Right has picked $v(\ell_j^r)$. Now, consider a trap-edge $e=\{t_k,\tau_k,d_j^r\}$ for some $k\in\segment{1}{n}$. By definition of the trap-edges, we have $\ind(\ell_j^r) \geq k$. Since $\ind(\ell_j^r)=i-1$, we have $k<i$, so Right has already killed $e$ by picking $\tau_k$ during round~$k$ of regular play.
        \item Finally, consider the case where regular play suggests picking $y_i^{\sf T}$ or $y_i^{\sf F}$ (the updated variable gadget is pictured in Figure \ref{fig:23-variable2}). Suppose that Right picks some $u \in V' \setminus \{y_i^{\sf T},y_i^{\sf F}\}$: we show that Left has a winning strategy after that move.
        We can assume that there exists $e\in E'_R$ of the form $e=\{u,u'\}$, otherwise Left has no forced move and wins in two moves by picking $y_i^{\sf T}$ for instance since $\{y_i^{\sf T},y_i^{\sf F}\},\{y_i^{\sf T},\zeta_i\} \in E'_L$ (we say that $(y_i^{\sf F},y_i^{\sf T},\zeta_i)$ is a blue $P_3$). Let us determine the set $E'_R$. We know which red edges were intact by the end of round~$i-1$ of regular play thanks to Claim~\ref{cla:regular0}, moreover some more red edges were killed during steps (1) and (2) of round~$i$: $\{x_i^{\sf T},w_i\}$, $\{x_i^{\sf F},w'_i\}$, $\{s_i,t_i\}$, and all $\{b_j^r,d_j^r\}$ such that $\ind(\ell_j^r)=i-1$ and Right has picked $v(\ell_j^r)$. All in all, we get:

        $$ E'_R = \{\{y_i^{\sf T},\omega_i\},\{y_i^{\sf F},\omega'_i\},\{\zeta_i,\zeta'_i\},\{\tau_i,\delta_i\}\} \cup \bigcup_{i+1 \leq k \leq n} G_k^R \cup \bigcup_{1 \leq j \leq m} \{\{b_j,d_j\}\} \cup \bigcup_{\substack{1 \leq j \leq m \\ 1 \leq r \leq 3 \\ \ind(\ell_j^r) \geq i}} \{\{b_j^r,d_j^r\}\}.$$

        \begin{itemize}
            \item[1)] Case $e \not\in \{\{y_i^{\sf T},\omega_i\},\{y_i^{\sf F},\omega'_i\},\{\zeta_i,\zeta'_i\},\{\tau_i,\delta_i\}\}$. This is where the trap-edges come into play. Left has triggered the trap-edges in $T_i$ by picking $t_i$ earlier. If Right had followed regular play, then she would have ended up picking $\tau_i$ after a sequence of forced moves, thus defusing these traps. However, she has not, and she will lose the game. Indeed, by definition of $T_i$, there exists the trap-edge $\{t_i,\tau_i,u'\}$. Left now picks $u'$ (forced move). Since Left had picked $t_i$ during step~(1), Right is forced to pick $\tau_i$. In turn, Left is forced to pick $\delta_i$ because $\{\tau_i,\delta_i\} \in E'_R$, then Right is forced to pick $\lambda_i$ because $\{\delta_i,\lambda_i\} \in E'_L$. Since $\lambda_i$ is in no red edge, Left has no forced move and wins using the blue $P_3$ $(y_i^{\sf F},y_i^{\sf T},\zeta_i)$ for instance.
            \item[2)] Case $e \in \{\{y_i^{\sf T},\omega_i\},\{y_i^{\sf F},\omega'_i\},\{\zeta_i,\zeta'_i\},\{\tau_i,\delta_i\}\}$ \textit{i.e.} $u \in \{\omega_i,\omega'_i,\zeta_i,\zeta'_i,\tau_i,\delta_i\}$. All options lead to a sequence of forced moves, at the end of which Left can use a blue $P_3$ to win:\\
            \begin{tikzcd}[row sep = tiny , sep = small]
                \textcolor{red}{u=\omega_i} \arrow{r}{\text{f}} & \textcolor{blue}{u'=y_i^{\sf T}}\,:\quad\text{Left wins with the blue $P_3$ $(y_i^{\sf F},y_i^{\sf T},\zeta_i)$;}
            \end{tikzcd}\\
            \begin{tikzcd}[row sep = tiny , sep = small]
                \textcolor{red}{u=\omega'_i} \arrow{r}{\text{f}} & \textcolor{blue}{u'=y_i^{\sf F}}\,:\quad\text{Left wins with the blue $P_3$ $(y_i^{\sf T},y_i^{\sf F},\zeta'_i)$;}
            \end{tikzcd}\\
            \begin{tikzcd}[row sep = tiny , sep = small]
                \textcolor{red}{u=\zeta_i} \arrow{r}{\text{f}} & \textcolor{blue}{u'=\zeta'_i}\,:\quad\text{Left wins with the blue $P_3$ $(y_i^{\sf F},\zeta'_i,\tau_i)$;}
            \end{tikzcd}\\
            \begin{tikzcd}[row sep = tiny , sep = small]
                \textcolor{red}{u=\zeta'_i} \arrow{r}{\text{f}} & \textcolor{blue}{u'=\zeta_i}\,:\quad\text{Left wins with the blue $P_3$ $(y_i^{\sf T},\zeta_i,\tau_i)$;}
            \end{tikzcd}\\
            \begin{tikzcd}[row sep = tiny , sep = small]
                \textcolor{red}{u=\tau_i} \arrow{r}{\text{f}} &
                \textcolor{blue}{u'=\delta_i} \arrow{r}{\text{f}} &
                \textcolor{red}{\lambda_i} \arrow{r}{} &
                \textcolor{blue}{y_i^{\sf T}}\,:\quad\text{Left wins with the blue $P_3$ $(y_i^{\sf F},y_i^{\sf T},\zeta_i)$;}
            \end{tikzcd}\\
            \begin{tikzcd}[row sep = tiny , sep = small]
                \textcolor{red}{u=\delta_i} \arrow{r}{\text{f}} & \textcolor{blue}{u'=\tau_i}\,:\quad\text{Left wins with the blue $P_3$ $(\zeta_i,\tau_i,\zeta'_i)$.}
            \end{tikzcd} \cqed\qedhere
        \end{itemize}

    \end{itemize}
\end{proof}

\begin{claim}\label{cla:regular3}
    Assume that Phase 1 is not over and that both players have followed regular play thus far. If Left is next to play, then it is optimal for Left to follow regular play with his next move.
\end{claim}

\begin{proof}[Proof of Claim~\ref{cla:regular3}]
    Note that, if there exists a clause $c_j$ such that Left has already picked $v(\ell_j^1)$, $v(\ell_j^2)$ and $v(\ell_j^3)$, then Left can follow regular play until the end of Phase 1 (making arbitrary choices for the remaining key vertices) and win the game. Indeed, Claim~\ref{cla:regular2} already tells us that Right has no interest in deviating from regular play, and Left will win during Phase 2 as we have seen in the proof of Claim~\ref{cla:regular1}. Therefore, we assume that there is no clause $c_j$ such that Left has already picked $v(\ell_j^1)$, $v(\ell_j^2)$ and $v(\ell_j^3)$. This assumption will only be used once in this proof.

    Let $\Ga'=(V',E'_L,E'_R)$ be the updated game, and let $i\in\segment{1}{n}$ be the current round of regular play. Since $E'_R$ forms a complete pairing of $(V',E'_R)$, Lemma \ref{lem:pairing} ensures that Left can play any move and still have a non-losing pairing strategy after that move. In particular, to establish optimality of the move(s) that regular play suggests for Left, it suffices to show that any other move for Left gives Right a drawing strategy. There are three types of moves for Left in regular play: forced move, greedy move, or ``decision move'' between $x_i^{\sf T}$ and $x_i^{\sf F}$.
    \begin{itemize}

        \item Let us start by checking that Left cannot win in one move. This is clear for the blue guide-edges and the link-edges, as regular play clearly has Right defending these threats. As for the butterfly-edges, Left never even picks one of their vertices during regular play. Now, consider a trap-edge $\{t_{i'},\tau_{i'},u\} \in T_{i'}$, and assume that Right has not yet picked $\tau_{i'}$ (otherwise that edge is already dead), which implies that $i' \geq i$. We claim that $u$ is also unpicked, so that Left cannot win in one move with that edge. To show this, we consider all three cases for the red edge $e$ that contains $u$, according to the definition of $T'_i$.
        \begin{itemize}[nolistsep,noitemsep]
            \item[--] First of all, suppose that $e \in G_k^R$ for some $k \in\segment{i'+1}{n}$. Since $k>i$, $u$ is clearly unpicked.
            \item[--] Next, suppose that $e=\{b_j,d_j\}$ for some $j\in\segment{1}{m}$. Since those vertices are never picked during regular play, $u$ is unpicked.
            \item[--] Finally, suppose that $e=\{b_j^r,d_j^r\}$ for some $j\in\segment{1}{m}$ and $r\in\segment{1}{3}$ such that $\ind(\ell_j^r) \geq i'$. Assume that $v(\ell_j^r)$ has been picked, otherwise $u$ is clearly unpicked. In particular, we necessarily have $\ind(\ell_j^r)=i'=i$. If Right has picked $v(\ell_j^r)$, then she will not play the greedy move $b_j^r$ until round~$i+1$ of regular play. If Left has picked $v(\ell_j^r)$, then he will not play the greedy move $d_j^r$ until step~(4) of the current round $i$ of regular play. Note that step~(4) has not been reached yet, otherwise step~(3) would have been over, so Right would have already picked $\tau_i=\tau_{i'}$.
        \end{itemize}
        We see that $u$ is unpicked in each case. In conclusion, Left cannot win in one move. In particular, every forced move suggested by regular play is indeed forced, as Right is clearly threatening to win in one move.

        \item Now, we consider the case where regular play suggests a greedy move $d_j^r$ during step~(4), which implies that $\ind(\ell_j^r)=i$ and that Left has picked $v(\ell_j^r)$ during the current round $i$ of regular play. We have just shown that Left cannot win in one move, moreover Right is clearly not threatening to win in one move either. Therefore, to show that this is a valid greedy move as per Lemma \ref{lem:greedy}, it suffices to show that the only edge (red or blue) containing $a_j^r$ in $\Ga'$ is the blue edge $\{a_j^r,d_j^r\}$. This is immediate, as the only edge (red or blue) containing $a_j^r$ in $\Ga$ is the blue edge $\{v(\ell_j^r),a_j^r,d_j^r\}$.
        
        \item Finally, we consider the case where regular play suggests picking $x_i^{\sf T}$ or $x_i^{\sf F}$. In particular, this means that round $i$ of regular play is just starting, so Claim~\ref{cla:regular0} applied to $i-1$ describes $\Ga'$ exactly. Suppose that Left picks some $u \in V' \setminus \{x_i^{\sf T},x_i^{\sf F}\}$ instead: we want to show that Right now has a drawing strategy. It is important to note that, by Claim~\ref{cla:regular0}, $\{x_i^{\sf T},x_i^{\sf F}\}$ is the only blue edge of size less than~$3$ in $\Ga'$, so that Right's next move is not forced. In an ideal scenario, Right would like to do the following (not necessarily in that order):
        \begin{itemize}[noitemsep,nolistsep]
            \item[--] Kill all remaining blue guide-edges, by picking $x_k^{\sf T}$ and $x_k^{\sf F}$ for all $k\in\segment{i}{n}$;
            \item[--] Kill all remaining trap-edges, by picking $t_k$ or $\tau_k$ for all $k\in\segment{i}{n}$;
            \item[--] Kill all remaining link-edges, by picking every unpicked $d_j^r$;
            \item[--] Kill all remaining butterfly-edges, by picking $b_j$ for all $j\in\segment{1}{m}$.
        \end{itemize}
        The idea is that all these vertices are in some intact red edge, so Right wants to pick all of them while forcing all of Left's moves.
        However, the vertex $u$ that Left has just picked may be one of these vertices, or the forced move associated with one of these vertices. Therefore, Right must adapt her plan to the vertex $u$.

        \begin{itemize}
        
        \item[1)] Case $u \in \bigcup_{1 \leq k \leq n} (V_k \setminus \{x_k^{\sf T},x_k^{\sf F}\})$. Let $k_0\in\segment{1}{n}$ (unique) such that $u \in V_{k_0}$.
            
            \begin{itemize}[noitemsep,nolistsep]

            \item[--] Right first kills all remaining trap-edges. If $u \in \{t_{k_0},s_{k_0},w_{k_0},w'_{k_0}\}$, then Right picks $\tau_{k_0}$ (which forces Left to pick $\delta_{k_0}$). Otherwise, Right picks $t_{k_0}$ (which forces Left to pick $s_{k_0}$). This way, she makes sure that $u$ does not contribute to a threat with the blue guide-edges in $G_{k_0}^L$. Then, for all $k\in\segment{i}{n}\setminus\quickset{k_0}$, she picks $t_k$ (which forces Left to pick $s_k$).
            
            \item[--] Right then kills all remaining butterfly-edges, by picking $b_j$ (which forces Left to pick $d_j$) for all $j\in\segment{1}{m}$. 
                
            \item[--] Right now kills all remaining link-edges, by picking every unpicked $d_j^r$ (which forces Left to pick $b_j^r)$. 
                
            \item[--] Finally, Right kills all remaining blue guide-edges. For all $k\in\segment{i}{n}\setminus\quickset{k_0}$, she picks $x_k^{\sf T}$ then $x_k^{\sf F}$ (which forces Left to pick $w_k$ then $w'_k$). Note that the move $x_k^{\sf F}$ was actually forced for Right because of the blue guide-edge $\{w_k,s_k,x_k^{\sf F}\}$. She does the same for $k_0$, except that she picks $x_{k_0}^{\sf F}$ first and $x_{k_0}^{\sf T}$ second if $u=w_{k_0}$. This way, the only move which does not force Left's answer is played last, when all blue edges are already dead.

            \end{itemize}

        \item[2)] Case $u \in (\bigcup_{1 \leq j \leq m} C_j) \setminus (\bigcup_{i \leq k \leq n} \{x_k^{\sf T},x_k^{\sf F},y_k^{\sf T},y_k^{\sf F}\})$. Let $j_0\in\segment{1}{m}$ (unique) such that $u \in C_{j_0}$.

            \begin{itemize}[noitemsep,nolistsep]
                
            \item[--] Right first kills all remaining trap-edges, by picking $t_k$ (which forces Left to pick $s_k$) for all $k\in\segment{i}{n}$.

            \item[--] Right then kills all remaining blue guide-edges, by picking $x_k^{\sf T}$ then $x_k^{\sf F}$ (which forces Left to pick $w_k$ then $w'_k$) for all $k\in\segment{i}{n}$.

            \item[--] Right now kills all remaining link-edges, by picking $y_k^{\sf T}$ and $y_k^{\sf F}$ (which forces Left to pick $\omega_k$ and $\omega'_k$ respectively) for all $k\in\segment{i}{n}$.

            \item[--] Finally, Right kills all remaining butterfly-edges. For this, she starts by picking $b_j$ (which forces Left to pick $d_j$) for all $j\in\segment{1}{m}\setminus\quickset{j_0}$. If $u \neq b_{j_0}$, then Right simply picks $b_{j_0}$ next to kill the last butterfly-edges. If $u=b_{j_0}$, then we use the assumption made at the beginning of the proof of this claim: one of $v(\ell_{j_0}^1)$, $v(\ell_{j_0}^2)$ or $v(\ell_{j_0}^3)$ has not been picked by Left. By Claim~\ref{cla:regular0}, this means that one of the destruction-edges $\{b_{j_0}^1,d_{j_0}^1\}$, $\{b_{j_0}^2,d_{j_0}^2\}$ and $\{b_{j_0}^3,d_{j_0}^3\}$ is intact, or that one of the butterfly-edges $\{b_{j_0},b_{j_0}^1,b_{j_0}^2\}$ and $\{b_{j_0},b_{j_0}^2,b_{j_0}^3\}$ is dead. As we have seen in the proof of Claim~\ref{cla:regular1}, this allows Right to hold on to a draw.

            \end{itemize}

        \item[3)] Case $u \in \bigcup_{i+1 \leq k \leq n} \{x_k^{\sf T},x_k^{\sf F}\}$. By symmetry, assume $u=x_{k_0}^{\sf T}$ for some $k_0\in\segment{i+1}{n}$.

            \begin{itemize}[noitemsep,nolistsep]
                
            \item[--] Right first kills all remaining butterfly-edges, by picking $b_j$ (which forces Left to pick $d_j$) for all $j\in\segment{1}{m}$. Left never threatens to win in one move with a trap-edge during the process, as $t_k$ and $\tau_k$ are both unpicked for all $k\in\segment{i}{n}$.

            \item[--] Right then kills all remaining link-edges, by picking every unpicked $d_j^r$ (which forces Left to pick $b_j^r$). Again, the trap-edges are not an issue.

            \item[--] Right now kills all remaining blue guide-edges in $G_k^L$ and trap-edges in $T_k$ for all $k\in\segment{i}{n}\setminus\quickset{k_0}$, by picking $x_k^{\sf T}$, $x_k^{\sf F}$ then $t_k$ (which forces Left to pick $w_k$, $w'_k$ then $s_k$ respectively). Note that, since Right has now picked $x_{k_0-1}^{\sf T}$ and $x_{k_0-1}^{\sf F}$, the blue guide-edges $\{x_{k_0}^{\sf T},x_{k_0}^{\sf F},x_{k_0-1}^{\sf T}\}$ and $\{x_{k_0}^{\sf T},x_{k_0}^{\sf F},x_{k_0-1}^{\sf F}\}$ are dead. Therefore, the updated variable gadget associated with $x_{k_0}$ and $y_{k_0}$ is as pictured in Figure \ref{fig:23-variable3}. The blue edge $\{t_{k_0},\tau_{k_0}\}$ comes from trap-edges in $T_{k_0}$ (for example, think of any trap-edge of the form $\{t_{k_0},\tau_{k_0},d_j\}$, since Left has picked $d_j$ while Right was killing the butterfly-edges).

            \begin{figure}[h]
                \centering
                \includegraphics[scale=.43]{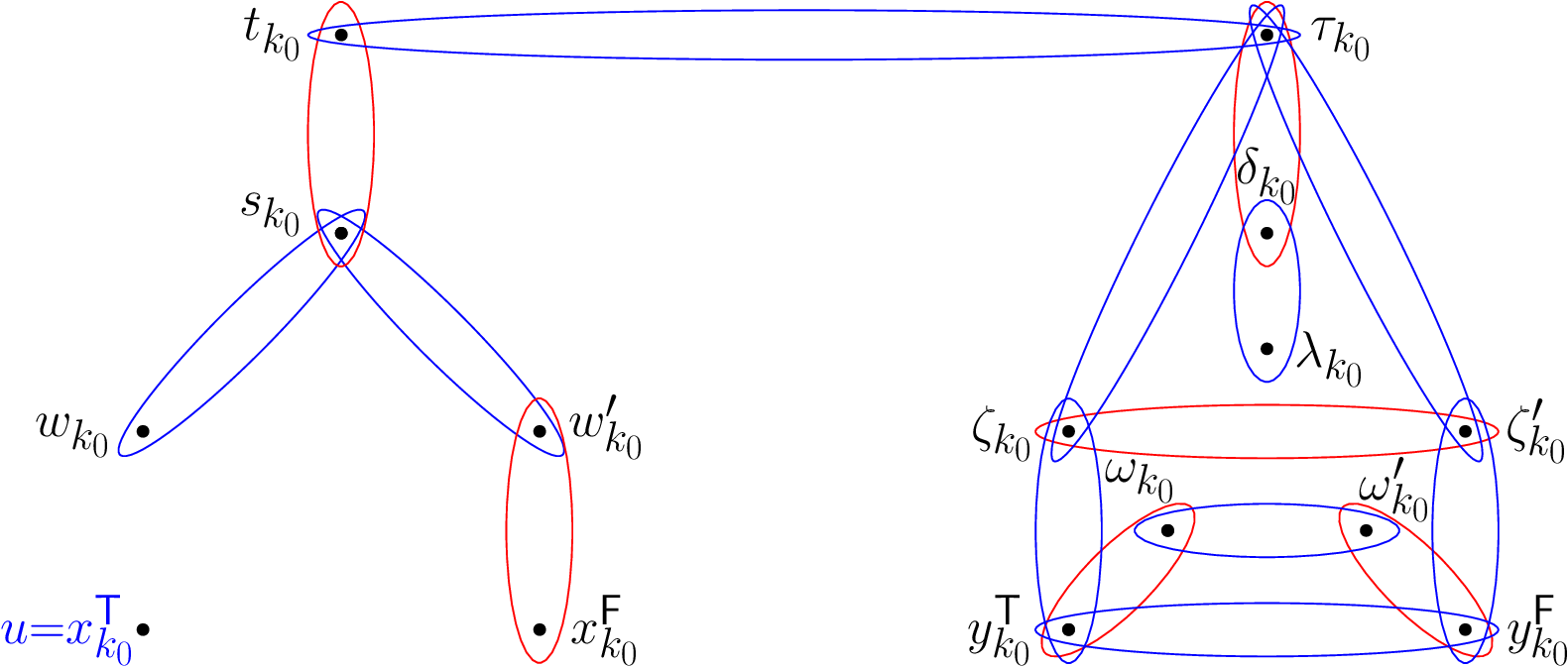}
                \caption{The updated variable gadget associated with $x_{k_0}$ and $y_{k_0}$ in the case $u=x_{k_0}^{\sf T}$ after Right has taken care of all the other variable gadgets.}\label{fig:23-variable3}
                \end{figure}
            
            \item[--] Finally, Right kills all remaining blue guide-edges in $G_{k_0}^L$ and trap-edges in $T_{k_0}$. She plays the following moves, which kill all remaining blue edges apart from the updated blue edge $\{w_{k_0},s_{k_0}\}$:
                \begin{center}
                \begin{tikzcd}[row sep = tiny , sep = small]
                    \arrow{r}{} & \textcolor{red}{w'_{k_0}} \arrow{r}{\text{f}} & \textcolor{blue}{x_{k_0}^{\sf F}} \arrow{r}{} & \textcolor{red}{y_{k_0}^{\sf T}} \arrow{r}{\text{f}} & \textcolor{blue}{\omega_{k_0}} \arrow{r}{\text{f}} & \textcolor{red}{\omega'_{k_0}} \arrow{r}{\text{f}} & \textcolor{blue}{y_{k_0}^{\sf F}} \arrow{r}{\text{f}} & \textcolor{red}{\zeta'_{k_0}} \arrow{r}{\text{f}} & \textcolor{blue}{\zeta_{k_0}} \arrow{r}{\text{f}} & \textcolor{red}{\tau_{k_0}} \arrow{r}{\text{f}} & \textcolor{blue}{\delta_{k_0}} \arrow{r}{\text{f}} & \textcolor{red}{\lambda_{k_0}}
                \end{tikzcd}
                \end{center}
                Right picks $w_{k_0}$ or $s_{k_0}$ with her next move to finish the job. \cqed\qedhere
            \end{itemize}
    \end{itemize}
\end{itemize}
\end{proof}

Putting claims \ref{cla:regular1}, \ref{cla:regular2} and \ref{cla:regular3} together, wee see that Left has a winning strategy on $\Ga$ as the first player if and only if Falsifier has a winning strategy for \QBF~on $\phi$ as the first player. We have $|V|=O(n+m)$, $|E_L|=O(n^2+nm)$ and $|E_R|=O(n+m)$, moreover the game $\Ga$ can be efficiently computed from the formula $\phi$, so this is indeed a polynomial-time reduction. This ends the proof of Theorem \ref{theo:32-new}.

\section{Conclusion}\label{section4}\strut
\indent When it comes to achievement positional games with blue edges of size at most~$p$ and red edges of size at most~$q$, we have shown that it is {\sf PSPACE}-complete to decide whether Left wins as the first player when $(p,q)=(3,2)$. In particular, the same is true for $(p,q)=(4,2)$. These two cases had only been proven to be {\NP}-hard in \cite{JonasFlorian}. The only remaining cases are $(p,q) \in \{(4,0),(4,1)\}$, which both correspond to Maker-Breaker games on hypergraphs of rank~$4$. 
While it is unlikely that an approach using general achievement positional games would help constructing a {\sf PSPACE}-hardness gadget for those due to Breaker's inability to create direct threats, the hardness of the case $(p,q)=(4,2)$ tightens the noose around this open problem.

We have also established {\sf PSPACE}-completeness for Maker-Maker games of rank~$4$. 
Previously, it was only known that Maker-Maker games of rank~$6$ were {\PSPACE}-complete.
Since Maker-Maker games of rank~$2$ are trivially solved in polynomial time, the case of rank~$3$ is the only one still open. An analogous {\sf PSPACE}-hardness gadget to that of the proof of Corollary \ref{coro:makermaker4} cannot be constructed for rank~$3$, as the case $(p,q)=(2,2)$ of achievement positional games is tractable. However, an approach using blue and red edges may still be useful. Indeed, edges of size~$3$ cannot be attributed to one particular player, but edges of size~$2$ can: they would initially be edges of size~$3$, which become strictly blue or red after some player has played inside them.

Another perspective would be to prove {\sf PSPACE}-completeness for $4$-uniform Maker-Maker games (actually, even 5-uniform Maker-Maker games are open in that regard). Our construction has rank~$4$, but it relies on the existence of smaller edges: one of size~$2$ and many of size~$3$. It would be easy to replace the edge of size~$2$ by two edges of size~$3$. However, getting rid of all edges of size~$3$ is not straightforward. With a $4$-uniform construction, Right would need to make two forced moves at the start of the game to ``activate'' the red edges of size~$2$, but then we would need a way to force the first two moves of Left as well: this is more difficult since, being the first player, Left has the advantage of making the first threat.

\bibliographystyle{biblio_style}
\bibliography{biblio}

\end{document}